\documentclass[11pt]{article}
\usepackage{graphicx} %

\usepackage{amssymb,amsthm,amsmath}
\usepackage{fullpage}
\usepackage{thm-restate}
\usepackage[ruled,linesnumbered]{algorithm2e}
\SetKw{Continue}{continue}
\SetKw{Break}{break}
 \DontPrintSemicolon 
\usepackage{multirow}
\usepackage{natbib}
\usepackage{breakcites}


\usepackage{subfig}
\usepackage[dvipsnames]{xcolor}
\usepackage[colorlinks,citecolor=blue,linkcolor=BrickRed]{hyperref}\usepackage[colorlinks,citecolor=blue,linkcolor=BrickRed]{hyperref}
\usepackage{fullpage}
\usepackage{graphicx}
\usepackage{enumitem}
\usepackage{thmtools,thm-restate}
\usepackage{wrapfig}
\usepackage{comment}
\usepackage[capitalize,noabbrev]{cleveref}
\usepackage{float}
\usepackage{soul}
\usepackage{thmtools} 
\usepackage{thm-restate}

\graphicspath{{figures/}}

\newtheorem{theorem}{Theorem}[section]
\newtheorem{corollary}[theorem]{Corollary}

\newtheorem{lemma}[theorem]{Lemma}

\newtheorem{definition}[theorem]{Definition}

\newtheorem{remark}[theorem]{Remark}

\newtheorem{fact}[theorem]{Fact}
\newtheorem*{claim*}{Claim}

\newif\ifFULL
\FULLtrue

\allowdisplaybreaks

\newcommand\restr[2]{{%
  \left.\kern-\nulldelimiterspace %
  #1 %
  \vphantom{\big|} %
  \right|_{#2} %
  }}

\newcommand{\argmax}{\mathrm{arg\,max}}

\renewcommand{\S}{\mathcal{S}}
\newcommand{\Atot}{\mathcal{A}_{\mathrm{tot}}}
\newcommand{\Stot}{\mathcal{S}_{\mathrm{tot}}}
\newcommand{\A}{\mathcal{A}}

\newcommand{\M}{\mathcal{M}}
\newcommand{\T}{\mathcal{T}}
\renewcommand{\P}{\mathbf{P}}
\newcommand{\eps}{\varepsilon}
\renewcommand{\epsilon}{\varepsilon}
\newcommand{\bA}{\mathbf{A}}
\newcommand{\policy}{\Pi}

\newcommand{\R}{\mathbb{R}}
\newcommand{\N}{\mathbb{N}}

\newcommand{\mdp}{\mathcal{M}}
\newcommand{\slack}{\Delta}

\newcommand{\otilde}{\Tilde{O}}
\newcommand{\poly}{\mathrm{poly}}

\newcommand{\defeq}{\stackrel{\mathrm{{\scriptscriptstyle def}}}{=}}

\newcommand{\E}{\mathbb{E}}

\newcommand{\vstar}{v_*}
\newcommand{\Det}{\mathrm{Det}}
\newcommand{\Rnd}{\mathrm{Rnd}}

\newcommand{\actionset}{\mathcal{A}}

\newcommand{\omegcurr}{\omega_{\mathrm{curr}}}

\title{Randomization for Faster Exact Optimization of\\ Discounted Markov Decision Processes}

\author{Andrei Graur \\ Stanford University \\\texttt{agraur@stanford.edu}
        \and
        Aaron Sidford \\ Stanford University \\ \texttt{sidford@stanford.edu}
        \and
        Ta-Wei Tu \\ Stanford University \\  \texttt{taweitu@stanford.edu}}
\date{}

\begin{document}

\maketitle

\pagenumbering{gobble}

\begin{abstract}
We provide faster algorithms for exactly solving discounted Markov Decision Processes (DMDPs) in strongly polynomial time. We obtain our results by efficiently reducing computing optimal values and policies in DMDPs to the easier tasks of policy evaluation and computing approximately optimal values. We provide both a straightforward deterministic reduction and a more efficient randomized variant that, together with advances in approximately solving DMDPs, yield our results.

\end{abstract}

\tableofcontents

\newpage
\pagenumbering{arabic}

\section{Introduction} \label{sec:intro}

Markov decision processes (MDPs) are ubiquitous in learning theory. They are a simple, but expressive, model for decision making under uncertainty that forms the basis for more complex learning theoretic problems.
Correspondingly, the problem of solving MDPs---and specifically, \emph{$\gamma$-discounted} MDPs (\emph{$\gamma$-DMDPs}\footnote{We note that DMDP is often also used to refer to deterministic Markov decision processes in the literature. } for short)---is a well-studied problem across computer science, operations research, and machine learning. 
The problem is foundational in reinforcement learning, which considers complex variations of it (see, e.g., \cite{DegrisSW06,sigaud2013markov,WeiXLGC17}), and the problem is closely related to canonical combinatorial optimization problems; for example, when $\gamma\rightarrow 1$ and transitions are deterministic, the problem encompasses the problem of computing a cycle of minimum-mean cost in a directed graph \cite{MadaniTZ09}. 

Given the prominence of the problem, there has been extensive research on efficiently solving $\gamma$-DMDPs \emph{approximately}. For example, convergence rates have been provided for a variety of methods, including value iteration, policy iteration \cite{howard1960dynamic,bellman1966dynamic,Puterman94,bertsekas1995neuro}, stochastic variants (in the context of MDPs) such as temporal learning \cite{Sutton88,TsitsiklisR02} and SARSA \cite{rummery1994line,SinghJLS00}, $Q$-learning \cite{WatkinsD92} (see \cite{sutton1998reinforcement} for a
survey on these methods), policy gradient \cite{Williams92,SuttonMSM99}, etc. There has also been extensive research on understanding the sample complexity of the problem~\cite{kakade2003sample,AzarMK12,AzarMK13,Wang17,SidfordWWYY18,Wainwright19,AgarwalKY20,JinS20,sidford2023variance,JinKSW24,JKSW25}.

Relatively less is known about \emph{exactly} solving $\gamma$-DMDPs, i.e., computing an optimal policy $\pi_*$
and the associated \emph{values} $v^{\M}_{\pi_*}$.
In this paper, we focus on the \emph{tabular} setting, where we have explicit access to the $\gamma$-DMDP $\M = (\S, \A, p, r)$ that consists of $\Stot \defeq |\S|$ states and a non-zero, non-uniform number of actions at each state which amount to a total of $\Atot \defeq |\A|$ actions (see \cref{sec:prelim} for the formal problem definition).\footnote{Some prior work on solving $\gamma$-DMDPs instead suppose that there is a fixed number of actions $A$ at every state and provide running times in terms of $\Stot A$ rather than $\Atot$. For uniformity of comparison and to reflect how results often generalize, we state running times for such prior work with $\Atot = \Stot A$ even if the work does not claim it.} Even in light of the above approximate algorithms, exactly solving $\gamma$-DMDPs in this setting is non-trivial, as there is no known target accuracy $\eps$ (in terms of the problem parameters) for policies and values that suffices to obtain an optimal solution. 

Nevertheless, strikingly, \cite{Ye05} provided what they call a 
``combinatorial interior point method (IPM)'' that solves $\gamma$-DMDPs exactly in $O(\Atot^4\log\frac{\Stot}{1-\gamma})$ time.
Additionally, the classic policy iteration method of \cite{howard1960dynamic} was shown to solve $\gamma$-DMDPs in $O(\frac{\Atot}{1-\gamma}\log(\frac{1}{1-\gamma}))$ iterations and can be implemented in $O((\frac{\Stot^\omega\Atot + \Stot\Atot^2}{1-\gamma})\log(\frac{1}{1-\gamma}))$ time~\cite{Ye11a,HansenMZ13,Scherrer13}, where $\omega < 2.372$ is the matrix multiplication constant \cite{AlmanDWXXZ25}. See \Cref{tab:prior_MDP_algs} for prior running times for exactly solving $\gamma$-DMDPs.

These exact $\gamma$-DMDP solvers are of interest for multiple reasons. First, they provide techniques for deducing that certain actions are not in an optimal policy and can thus be discarded. Second, there have been results for solving MDPs in stochastic settings where exact MDP solvers are used 
\cite{BrunskillLLLR09,XuT20,RosenbergM21}
(though perhaps this exactness requirement could be relaxed). 
Third, they efficiently solve a special class of linear programs in strongly polynomial time (see the discussion on exact algorithms for $\gamma$-DMDPs and LPs in \Cref{sec:related_work}). 
Moreover, interestingly, while policy iteration can be viewed as a type of simplex method, particular simplex methods are only known to run in weakly polynomial time~\cite{KelnerS06}.

Correspondingly, the goal of this work is to obtain more efficient exact $\gamma$-DMDP solvers. We seek both faster algorithms and general frameworks which enable these algorithms.

\subsection{Our Results}
\label{sec:results}

In this paper we provide more efficient %
algorithms for solving $\gamma$-DMDPs \emph{exactly}. We provide deterministic algorithms with running times that improve upon the prior state of the art and randomized algorithms that further improve these running times by a factor of $\tilde{\Omega}(\Atot/\Stot)$.\footnote{Throughout this paper, we use $\tilde{O}(\cdot)$ and $\tilde{\Omega}(\cdot)$ to hide $\poly\log(\Atot)$ factors (but not $\poly\log(1/(1-\gamma))$).} 
For a comparison between prior running times and our 
improved running times, see \Cref{tab:aproxDMDPsolvers,tab:prior_MDP_algs}. 

Perhaps more importantly, we obtain these results through a straightforward framework (\Cref{alg:reduction}) that reduces solving $\gamma$-DMDPs to the easier tasks of \emph{policy evaluation}\footnote{Solving the $\gamma$-DMDP that contains only the actions of a policy $\pi$, by definition, computes $v_\pi$.} (i.e., computing $v^{\M}_{\pi}$ for a policy $\pi$) which is solvable in $O(\Stot^\omega)$ time, and, essentially, solving a $\gamma$-DMDP \emph{approximately}. 

Our framework is broadly inspired by \cite{Ye05} which provided an IPM that discards \textit{suboptimal} (i.e., not part of any optimal policy) actions over time.
At a high level, our algorithm proceeds similarly but deviates in exactly how it reasons about suboptimal actions and replaces using an IPM with an arbitrary approximate algorithm for solving $\gamma$-DMDPs. We essentially apply an analysis similar to \cite{Ye11a,Scherrer13} inside this general framework to perform a straightforward analysis of the iteration count. Ultimately, this allows us to obtain faster deterministic algorithms for solving $\gamma$-DMDPs by leveraging state-of-the-art approximate solvers (and essentially recovers the analysis of \cite{Ye05}). 
Additionally, we show that it is possible to modify this general framework to incorporate randomization; we show that a randomized variant of the framework allows us to prove that more actions can be discarded (in expectation) and thereby to decrease the number of iterations.\footnote{It is conceivable that similar improvements for deterministic algorithms could be obtained by incorporating insights of \cite{Ye05} into recent IPM results. It is also conceivable that one can incorporate our randomization insights within the framework of \cite{Ye05}. It may be useful that we obtain black-box reductions without appealing to IPMs.}

\paragraph{From Approximate Values to Exact Policies.}
Our notion of approximately solving $\gamma$-DMDPs is a natural scale-invariant definition of computing $\delta$-approximate values as defined in \Cref{def:apx_values} below (where $\delta$ is set in the reductions). 
This definition simply scales the accuracy of the values computed with $\|r\|_{\infty}$, where $r \in \R^{\actionset}$ is the vector whose entries are the rewards of each action. 

\begin{definition}[$\gamma$-DMDP Approximations]
\label[definition]{def:apx_values}
We call $v \in \R^\S$ \emph{$\epsilon$-optimal}
values of $\gamma$-DMDP $\M = (\S, \actionset, p, r)$ for accuracy $\epsilon > 0$ if $\|v - \vstar^{\M}\|_{\infty} \le \epsilon$ and we call a policy $\pi$ of $\M$ \emph{$\epsilon$-optimal} if $\|v_{\pi}^{\M} - \vstar^\M\|_{\infty} \le \epsilon$, where $\vstar^\M$ is the \emph{optimal values} of $\M$ (see \cref{sec:prelim}). We call an algorithm a \emph{$\delta$-approximate $\gamma$-DMDP
solver} for accuracy $\delta \in (0,1)$ if for input $\gamma$-DMDP $\M = (\S, \actionset, p, r)$ it outputs $\delta \|r\|_{\infty}$-optimal values
$v$ of $\M$.
\label{def:apx-algo}
\end{definition}

We provide deterministic and randomized reductions to approximate $\gamma$-DMDP solvers. To present these results we let
$\T_{\Det}(\Stot, \Atot, \delta_{\gamma})$ be the running time of a deterministic $\delta_{\gamma}$-approximate $\gamma$-DMDP algorithm on input $\M = (\S, \actionset, p, r)$ and let $\T_{\Rnd}(\Stot, \Atot, \delta_{\gamma})$ be the expected running time of such a randomized algorithm.
See \Cref{tab:aproxDMDPsolvers} for the running times for exactly solving $\gamma$-DMDPs that we obtain using our reductions and different approximate $\gamma$-DMDP algorithms.

Our first reduction is deterministic.
At a high level, it shows that given an arbitrary policy $\pi$, as long as it is not optimal, an approximate $\gamma$-DMDP solve (plus some additional evaluation of policies) suffices to detect one suboptimal action (specifically, from $\pi$) in the $\gamma$-DMDP.
Discarding that action and repeating thus gives us an $O(\Atot)$-iteration reduction.

\begin{restatable}[Deterministic Reduction]{theorem}{DetReduction}
  There is a deterministic algorithm that solves a $\gamma$-DMDP $\M = (\S, \actionset, p, r)$ in $O(\Stot^{\omega}\Atot + \Stot\Atot^2) + O(\Atot) \cdot \T_{\Det}(\Stot, \Atot, \Theta(1-\gamma))$ time.
  \label{thm:det-reduction}
\end{restatable}

State-of-the-art deterministic $\delta_{\gamma}$-approximate $\gamma$-DMDP solver include an $\tilde{O}(\Atot^{\omegcurr}\log(\frac{1}{(1-\gamma)\delta_{\gamma}}))$ time\footnote{The actual running time of \cite{Brand20} is $\tilde{O}((\Atot^{\omega} + \Atot^{2.5-\alpha/2+o(1)} + \Atot^{2+1/6+o(1)})\log(\frac{1}{(1-\gamma)\delta_{\gamma}}))$, where $\omega$ is the matrix multiplication constant and $\alpha$ is the dual matrix exponent (i.e., largest $\mu$ so that multiplying a $n \times n^{\mu}$ matrix by a $n^{\mu} \times n$ matrix can be done in $O(n^{2 + \epsilon})$ time for every $\epsilon>0$). For the current values of $\omega = \omega_{\mathrm{curr}} \approx 2.371339$~\cite{AlmanDWXXZ25} and $\alpha = \alpha_{\mathrm{curr}} \ge 0.321334$~\cite{Gall24,WilliamsXXZ24}, the running time of \cite{Brand20} simplifies to $\tilde{O}(\Atot^{\omega_{\mathrm{curr}}}\log(\frac{1}{(1-\gamma)\delta_{\gamma}}))$.}
 IPM-based algorithm by \cite{Brand20}\footnote{\cite{CohenLS19,LeeSZ19} obtain similar running times but their algorithms are randomized. \cite{Jiang0WZ21} also obtains the same running time for the best-known value of $\omega$ and improves by a polynomial factor over these works if $\omega = 2$.} (using a reduction in \cite{sidford2023variance}), where $\omegcurr \approx 2.371339$ \cite{DuanWZ23,AlmanW21,Gall24,WilliamsXXZ24,AlmanDWXXZ25}
 is the current value of the matrix multiplication constant,
 and the
 well-known value iteration algorithm of \cite{bellman1966dynamic} that  runs in $O(\frac{\Stot\Atot}{1-\gamma}\log(\frac{1}{(1-\gamma)\delta_{\gamma}}))$ time (see, e.g., \cite{JinKSW24}). 
Instantiating \cref{thm:det-reduction} with these algorithms implies the following. 

\begin{corollary}
  There is a deterministic algorithm that solves a $\gamma$-DMDP $\M = (\S, \actionset, p, r)$ in the minimum of 
  $O(\Stot^\omega \Atot + \frac{\Stot\Atot^2}{1-\gamma}\log(\frac{1}{1-\gamma}))$ and
  $\tilde{O}(\Atot^{\omegcurr+1}\log(\frac{1}{1-\gamma}))$ time.
  \label{cor:det}
\end{corollary}

\paragraph{Randomization for Faster Algorithms.}

While our deterministic reduction removes the reliance of \cite{Ye05} on using a specific IPM, our best analysis only bounds that it proceeds in a similar
$\Theta(\Atot)$ iterations; in our worst case analysis it may only detect one suboptimal action per approximate solve.
We improve this bound by showing that instead of detecting suboptimal actions based on an \emph{arbitrary} policy $\pi$,
if we pick $\pi$ \emph{uniformly at random} at each iteration, then asymptotically more actions can be discarded in expectation.
In particular, this improves the expected iteration count to $\tilde{O}(\Stot)$ which is lower than $O(\Atot)$ whenever there are sufficiently many actions per state.

\begin{restatable}[Randomized Reduction]{theorem}{RandReduction}
  There is a randomized algorithm that solves a $\gamma$-DMDP $\M = (\S, \actionset, p, r)$ in $\tilde{O}(\Stot^{\omega+1} + \Stot^2\Atot) + \tilde{O}(\Stot) \cdot \T_{\Rnd}(\Stot, \Atot, \Theta(1-\gamma))$ expected time.
  \label{thm:rand-reduction}
\end{restatable}

\begin{table}
    \centering
    \begin{tabular}{|c|c|c|}
         \hline
         Paper & Method & Running time \\
         \hline
         \cite{Ye05} & Combinatorial IPM & $O(\Atot^4\log(\frac{\Stot}{1-\gamma}))$ \\
         \hline
         \cite{Ye11a} & Simplex PI & $O((\frac{\Stot^{\omega+1}\Atot + \Stot^2\Atot^2}{1-\gamma})\log(\frac{1}{1-\gamma}))$ \\
         \hline
         \cite{Scherrer13} & Howard's PI & $O((\frac{\Stot^\omega\Atot + \Stot\Atot^2}{1-\gamma})\log(\frac{1}{1-\gamma}))$ \\
         \hline
    \end{tabular}
    \caption{Prior running times for solving $\gamma$-DMDPs exactly.}
    \label{tab:prior_MDP_algs}
\end{table}

\begin{table}
\centering
\begin{tabular}{|p{4.2cm}|c|p{2.3cm}|}
\hline
\centering Approximate solver & Implied running time & Deterministic or randomized \\\hline
\centering\cite{bellman1966dynamic} & $O(\Stot^\omega\Atot + \frac{\Stot\Atot^2}{1-\gamma}\log(\frac{1}{1-\gamma}))$ & Deterministic \\\hline
\centering\cite{Brand20} & $\tilde{O}(\Atot^{\omegcurr+1}\log(\frac{1}{1-\gamma}))$ & Deterministic \\\hline
\centering\cite{Brand20} & $\tilde{O}(\Stot\Atot^{\omegcurr}\log(\frac{1}{1-\gamma}))$ & Randomized \\\hline
\centering\cite{BrandLLSS0W21} & $\tilde{O}((\Stot^2\Atot + \Stot^{3.5})\log(\frac{1}{1-\gamma}))$ & Randomized \\\hline
\centering\cite{JKSW25}
& $\tilde{O}(\Stot^{\omega+1} + (\Stot^2\Atot + \frac{\Stot^{1.5}\Atot}{1-\gamma})\log^{O(1)}(\frac{1}{1-\gamma}))$ & Randomized \\\hline
\end{tabular}
\caption{Improved running times for solving $\gamma$-DMDPs exactly.}
\label{tab:aproxDMDPsolvers}
\end{table}

Note that \cref{thm:rand-reduction} requires the approximate $\gamma$-DMDP solver to have an expected running time instead of, e.g., succeeding with constant probability. 
However, in the following \Cref{lemma:mc-to-lv-intro} (proved in \cref{sec:monte-carlo-to-las-vegas}) we show that we can apply fairly well-known techniques and properties of $\gamma$-DMDPs to efficiently convert an approximate $\gamma$-DMDP algorithm that succeeds with constant probability into one with an expected running time, with a limited loss in approximation.

\begin{restatable}{lemma}{MCtoLV}
    Suppose there is a randomized algorithm that on input $\gamma$-DMDP $\M^\prime=(\S^\prime,\actionset^\prime,p^\prime,r^\prime)$ outputs $\delta^\prime \|r^\prime\|_{\infty}$-optimal values $v$ of $\M$ in time $\T_{\mathrm{MC}}(\Stot^\prime, \Atot^\prime, \delta^\prime)$ with constant probability.
Then, there is a randomized $\delta$-approximate $\gamma$-DMDP algorithm that on input $\gamma$-DMDP $\M=(\S, \actionset, p, r)$ runs in expected time $\T_{\Rnd}(\Stot, \Atot, \delta) = O(\Stot\Atot + \T_{\mathrm{MC}}(\Stot, \Atot, \Theta(\delta \cdot (1-\gamma)))$.
\label[lemma]{lemma:mc-to-lv-intro}
\end{restatable}

For randomized Monte Carlo $\delta_{\gamma}$-approximate $\gamma$-DMDPs solvers, 
the state of the art includes 
an
$\tilde{O}((\Stot\Atot + \frac{\sqrt{\Stot}\Atot}{1-\gamma}) \log^{O(1)}\frac{1}{(1-\gamma)\delta_{\gamma}})$ time\footnote{We note that \cite{JKSW25} only claimed a running time of $\tilde{O}((\Stot\Atot + \frac{\sqrt{\Stot}\Atot}{1-\gamma}) \log^{O(1)}\frac{1}{(1-\gamma)\delta_{\gamma}})$ when $\Stot\Atot\leq \Atot(1-\gamma)^2$. However, when $\Stot\Atot > \Atot/(1-\gamma)^2$, the $\tilde{O}((\Stot\Atot + \Atot/(1-\gamma)^2)\log^{O(1)}\frac{1}{(1-\gamma)\delta_{\gamma}})$ running of \cite{JinKSW24} would also be bounded by $\tilde{O}((\Stot\Atot + \frac{\sqrt{\Stot}\Atot}{1-\gamma}) \log^{O(1)}\frac{1}{(1-\gamma)\delta_{\gamma}})$. We additionally note that this running time is better than the standard value iteration \cite{bellman1966dynamic} in all regimes of parameters.}
sampling-based value iteration algorithm \cite{sidford2023variance,JinKSW24,JKSW25} and an
$\tilde{O}((\Stot\Atot + \Stot^{2.5})\log(\frac{1}{(1-\gamma)\delta_{\gamma}}))$ time IPM-based algorithm \cite{BrandLLSS0W21}. 
Combining \cref{lemma:mc-to-lv-intro} with these algorithms (as well as that in \cite{Brand20}) and plugging them into \cref{thm:rand-reduction}, we obtain the following running times for exactly solving $\gamma$-DMDPs.

\begin{corollary}
  There is a randomized algorithm that solves an
  input $\gamma$-DMDP $\M = (\S, \actionset, p, r)$ in the minimum of $\tilde{O}(\Stot^{\omega+1} + (\Stot^2\Atot + \frac{\Stot^{1.5}\Atot}{1-\gamma})\log^{O(1)}(\frac{1}{1-\gamma}))$, $\tilde{O}((\Stot^2\Atot + \Stot^{3.5})\log(\frac{1}{1-\gamma}))$, and $\tilde{O}(\Stot\Atot^{\omegcurr} \allowbreak \log(\frac{1}{1-\gamma}))$ expected time.  
  \label{cor:main}
\end{corollary}

Lastly, in \cref{appendix:random-policy}, we show that instantiating \cref{thm:rand-reduction} instead with the classic policy iteration algorithm of \cite{howard1960dynamic} results in a simple randomized policy iteration with a running time that improves upon e.g., \cite{Scherrer13}, but does not improve upon applying our framework with the (fastest) sampling-based value iteration methods of \cite{sidford2023variance,JinKSW24,JKSW25}.

\subsection{Additional Related Work}
\label{sec:related_work}

\paragraph{Approximate Algorithms for $\gamma$-DMDPs.} 

There have been numerous advances 
in obtaining fast approximate algorithms for $\gamma$-DMDPs in different parameter regimes.
Beyond classic algorithms such as policy and value iterations~\cite{bertsekas2012dynamic,bellman1966dynamic,howard1960dynamic,tseng1990solving,LittmanDK95}, there are several works that give sampling-based methods for $\gamma$-DMDP, some of them in the context of reinforcement learning~\cite{LattimoreH12,AzarMK12,AzarMK13,sidford2023variance,JinKSW24,JKSW25}.
Since $\gamma$-DMDPs are instances of linear programs, they can also be solved to high accuracy using weakly polynomial time linear programming (LP) algorithms (e.g., using a reduction of \cite{sidford2023variance}).
High-accuracy algorithms for linear programs have been extensively studied and there has been recent progress in improving their running times \cite{LeeS14,LeeS15,CohenLS19,Brand20,BrandLLSS0W21,Jiang0WZ21} some of which yield several state-of-the-arts for $\gamma$-DMDPs.

\paragraph{Exact Algorithms for $\gamma$-DMDPs and LPs.}
Our work on exact algorithms for $\gamma$-DMDPs builds on a line of work on providing strongly polynomial time
algorithms for solving $\gamma$-DMDPs when the discount factor $\gamma$ is a constant. 
\cite{Ye05} provided the first analysis of such an algorithm for this variant of the problem, followed by~\cite{Ye11a,HansenMZ13,Scherrer13}. 
Strongly polynomial algorithms can be defined for different models of computation.
In the \emph{real RAM} model, an algorithm runs in strongly polynomial time if the number of arithmetic operations it performs does not depend on the bit complexity of the inputs.
In the stricter \emph{Turing model}, on the other hand, it is additionally required that the \emph{space complexity} of the algorithm is bounded by $\mathrm{poly}(L)$, where $L$ is the total bit complexity of the inputs.\footnote{One classic example is that one can compute $2^{2^n}$ using $n$ multiplications. However, the space complexity of such an algorithm would be $2^n$, so it is considered strongly polynomial in real RAM but not the Turing model.}
In most of the paper we do not discuss the space complexity;
however, in \Cref{rem:strongly_poly}, we briefly discuss our framework's performance in the Turing model. %

Related to work on strongly polynomial algorithms for 
solving $\gamma$-DMDPs with constant $\gamma$ is a long sequence of works~\cite{Tardos86,VavasisY96,DadushNV20,DadushHNV20} 
that obtain exact %
algorithms for LPs whose running times only depend on certain condition numbers of the constraint matrix $\bA$ of the linear program (e.g., $\bar{\chi}_{\bA}$ or $\bar{\chi}^*_{\bA})$.
These algorithms are thus strongly polynomial for certain classes of suitably-conditioned LPs, which
constitutes progress towards the outstanding open problem of obtaining a strongly polynomial time algorithm for (general) LPs in the Turing model \cite{Megiddo83} (see \cite{DadushKNOV24} for further discussion). 
However, when $\mathbf{A}$ is the constraint matrix in the linear programming formulation of solving $\gamma$-DMDPs, the quantities $\chi_{\bA}$ and $\chi_{\bA}^*$ still depend on the bit complexity of the transition probabilities of the $\gamma$-DMDP. 
Hence, one cannot apply the results in~\cite{Tardos86,VavasisY96,DadushNV20,DadushHNV20} black-box to obtain strongly polynomial time algorithms for $\gamma$-DMDP, even for bounded values of $\gamma$ (though there might be ways to closely follow these frameworks and obtain such algorithms). 
Additionally, there is a line of work \cite{DadushKNOV24,AllamigeonDLNV25,DadushMNV26} on designing IPMs whose running time is tightly related (i.e., up to polynomial factors) to a quantity that \cite{AllamigeonDLNV25} call the \emph{straight line complexity (SLC)} of the linear program. Independently of this paper, \cite{bentoPC26} 
claims that the SLC of the LP formulation of solving a $\gamma$-DMDP is $O(\poly(\Atot) \log(1/(1-\gamma)))$.

We note that some of these works, e.g.,  \cite{Tardos86,DadushNV20}, analogously to this paper, give a reduction from exact to approximate LP algorithms. In fact, Tardos' framework, when applied to $\gamma$-DMDPs, proceeds similarly to our deterministic framework;
    however, we leverage the particular optimality properties specific to $\gamma$-DMDPs to obtain %
    our particular bound on the accuracy required for the approximate solves.

\paragraph{Deterministic $\gamma$-DMDPs.}

$\gamma$-DMDP instances where $p(s,a) \in \{0, 1\}$ for all state-action pairs $(s,a)$ (i.e., transitions to a new state are deterministic) are called deterministic $\gamma$-DMDPs and are well-studied. 
In particular, strongly polynomial time algorithms that do not depend on the value of $\gamma$ are known for deterministic $\gamma$-DMDPs~\cite{MadaniTZ09,Karczmarz22}.
Simplex methods are also known to be 
strongly polynomial time algorithms on deterministic $\gamma$-DMDPs \cite{PostY13,HansenKZ14}, yet such variants have worse running times than the specialized methods of \cite{MadaniTZ09,Karczmarz22}.

\paragraph{Simplex Methods and $\gamma$-DMDPs.}
Given the close relationship between policy iteration algorithms and simplex methods for $\gamma$-DMDPs (for instance, the version of policy iteration that only switches the action at a single state every time \cite{Ye11a} is one type of simplex method for $\gamma$-DMDPs), 
we note that there have been several results on proving lower bounds for certain variants of the simplex methods. 
In particular, \cite{MelekopoglouC94,Fearnley10,FriedmannHZ11,hansen2012worst,AvisF17,DisserFH23} show that several (deterministic or randomized) variants of the simplex method applied for $\gamma$-DMDPs require subexponential (or exponential) time when $\gamma = 1$. 
Note that in the rest of the paper we only consider the setting where $\gamma$ is bounded away for $1$, and thus the lower bounds in the works mentioned do not contradict our results.

\subsection{Paper Organization}

The rest of the paper is organized as follows.
In \cref{sec:prelim}, we establish the necessary notation and preliminaries.
In \cref{sec:reduction}, we present our framework for reducing the problem of exactly solving a $\gamma$-DMDP to approximately solving several $\gamma$-DMDPs and prove our deterministic reduction (\Cref{thm:det-reduction}). 
In \Cref{sec:rand_reduction}, we instantiate our framework to obtain our randomized reduction (\Cref{thm:rand-reduction}). 
In \cref{appendix:random-policy}, we again use the framework and the randomized analysis to present our simple randomized algorithm based on the policy iteration algorithm of \cite{howard1960dynamic}. 
In \cref{sec:monte-carlo-to-las-vegas} we provide a fairly straightforward proof of \cref{lemma:mc-to-lv-intro}.

\section{Preliminaries} \label{sec:prelim}

\paragraph{General Notation.}
We write $u \leq v$ (respectively, $u \geq v$) for $u, v \in \R^d$ to indicate $u_i \leq v_i$ (respectively, $u_i \geq v_i$) for all $i \in [d]$. 

\paragraph{$\gamma$-DMDPs.} This paper considers the problem of solving or optimizing a $\gamma$-DMDP for $\gamma \in (0, 1)$. We specify a $\gamma$-DMDP by a tuple 
$\M = (\S, \actionset, p, r)$, where $\S$ is a non-empty finite set of \emph{states}, $\actionset = \{(s, a) \mid s \in \S\}$ is a finite set of \emph{state-action pairs}, $p: \actionset \to \slack^\S$ for $\slack^\S \defeq \{q\in \R^\S_{\geq 0} \mid \|q\|_1 = 1\}$ denotes \emph{transition probabilities}, and $r \in \R^\actionset$ denotes \emph{rewards}. We call $\A_s \defeq \{a \mid (s,a) \in \actionset\}$ the \emph{actions} at $s \in \S$ and assume each $\A_s \neq \emptyset$. 

An \textit{agent} interacts with a $\gamma$-DMDP in time-steps $t = 0,1,\ldots$. If the agent takes action $a_t \in \A_{s_t}$ at state $s_t \in \S$ at time-step $t$ then the agent receives reward $r_{(s_t,a_t)} \in \R$ and randomly \emph{transitions} or \emph{moves} to state $s_{t+1} \in \S$ sampled independently from distribution $p(s_t,a_t)$ for time-step $t+1$, i.e., $\Pr[s_{t+1} = \bar{s} \mid s_t] = p(s_t,a_t)_{\bar{s}}$.

In a $\gamma$-DMDP we optimize over (deterministic stationary) \textit{policies} $\pi$ which map states $s \in \S$ to actions at that state, i.e.,
$\pi(s) \in \A_s$. 
The value of a policy $\pi$ at state $s \in \S$, denoted by $[v^{\M}_\pi]_s$, is the expected $\gamma$-discounted reward received by an agent following $\pi$ starting from state $s$, i.e., 
\begin{equation}
\label{eq:val_vec}
[v^\M_\pi]_s \defeq \E \left[\sum_{t = 0}^{\infty} \gamma^t r_{s_t,\pi(s_t)}\right]
\text{ where }
s_0 = s
\text{ and }
\Pr[s_{t+1} = \bar{s} | s_{t},\ldots,s_0] = p(s_t, \pi(s_t))_{\bar{s}},\;\forall
t \in \N
\,,\footnote{For notational convenience, we may omit parentheses in subscripts, e.g., we may write $r_{s,a}$ instead of $r_{(s,a)}$.}
\end{equation}
where the transitions are drawn independently from each other.
Let $[v^\M_*]_s \defeq \max_{\pi} [v_{\pi}^\M]_s$ be the \emph{optimal values}.
It is known that $v^\M_* = v^\M_{\pi_*}$ for some policy $\pi_*$ (see, e.g., \cite{Puterman94}); we call such a policy $\pi_*$ \emph{optimal}.

\paragraph{Bellman Operator}
To reason about policies and values,  
we define the \emph{Bellman on-policy operator} $\T^\M_{\pi}: \R^\S \to \R^\S$ with respect to a policy $\pi$ by $[\T^\M_{\pi}(v)]_s \defeq r_{s,\pi(s)} + \gamma \cdot \langle p(s,\pi(s)), v \rangle$, and the \emph{Bellman value operator} $T^\M_*: \R^\S \to \R^\S$ that maximizes over all actions instead of a fixed policy, i.e., $[T^\M_*(v)]_s \defeq \max_{a\in \A_s} r_{s,a} + \gamma \cdot \langle p(s,a), v\rangle$.
It is known that both $\T^\M_{\pi}$ and $T^\M_*$ are contractions and we use the following folklore contraction bounds for $\gamma$-DMDPs throughout the paper.

\begin{restatable}[see, e.g., {\cite[Lemmas 3.4 and 3.5]{sidford2023variance}}]{lemma}{ValueBound}
  Let $\M = (\S, \actionset, p, r)$ be a $\gamma$-DMDP.
  The unique fixpoint of $\T^\M_*$ is $v_*^\M$, and for any values $u, v \in \R^\S$ it holds that
  \[
    \|\T_{*}^{\mdp}(u) - \T_{*}^{\mdp}(v)\|_{\infty} \leq \gamma \left\|u - v\right\|_{\infty}\quad\text{and}\quad
    \|u - v^{\M}_{*}\|_{\infty} \leq \frac{1}{1-\gamma}\left\|u - \T^{\M}_{*}(u)\right\|_{\infty} \,.
  \]
  \label[lemma]{lemma:value-bound}
\end{restatable}

\ifdefined\ARXIV
\begin{proof}
  The first inequality is \cite[Lemma 3.4]{sidford2023variance}.
  The second inequality follows from essentially the same reasoning as in the proof  \cite[Lemma 3.5]{sidford2023variance} (and is a general property of contractions, not just for $\T_*$).
  Applying the first inequality and $\T_*(v_*) = v_*$ yields 
  \begin{align*}
    \|\T_*(u) - v_*\|_{\infty}
    &= \|\T_*(u) - \T_*(v_*)\|_{\infty} = \|\T_*(u) - \T_*(\T_*(u)) + \T_*(\T_*(u)) - \T_*(v_*)\|_{\infty} \\
    &\leq \|\T_*(u) - \T_*(\T_*(u))\|_{\infty} + \|\T_*(\T_*(u)) - \T_*(v_*)\|_{\infty} \\
    &\leq \gamma \|u - \T_*(u)\|_{\infty} + \gamma\|\T_*(u) - v_*\|_{\infty}\,.
  \end{align*}
  This  implies $\|\T_*(u) - v_*\|_{\infty} \leq \frac{\gamma}{1-\gamma}\|u - \T_*(u)\|_{\infty}$ and therefore
  \[
    \|u - v_*\|_{\infty} \leq \|u - \T_*(u)\|_{\infty} + \|\T_*(u) - v_*\|_{\infty} \leq \frac{1}{1-\gamma}\|u - \T_*(u)\|_{\infty}.
    \qedhere
  \]
\end{proof}
\fi

Note that the above inequalities apply to $\T^\M_{\pi}$ and $v^\M_{\pi}$ as well by simply considering the $\gamma$-DMDP with only actions in $\pi$, i.e., 
\begin{equation}
\label{eq:gen_contract}
    \|\T_{\pi}^{\mdp}(u) - \T_{\pi}^{\mdp}(v)\|_{\infty} \leq \gamma \left\|u - v\right\|_{\infty}\quad\text{and}\quad
    \|u - v^{\M}_{\pi}\|_{\infty} \leq \frac{1}{1-\gamma}\left\|u - \T^{\M}_{\pi}(u)\right\|_{\infty} \,.
\end{equation}

\paragraph{Advantage Function Values.}
To reason about the suboptimality of values and actions, we define, 
for values $v\in \R^S$, the \emph{advantage function values} $\slack^\M(v) \in \R^{\actionset}$ with respect to $v$ as $\slack^\M(v)_{s, a} \defeq (r_{s,a}+\gamma \cdot \langle p(s,a), v \rangle) - v_s$ for every $(s, a) \in \actionset$.\footnote{This coincides with the natural notion of \emph{slack} when expressing $\gamma$-DMDPs as a linear program.} 
For brevity, throughout the paper, we often use \emph{advantage(s)} to refer to advantage function value(s). 
Note that $\slack^\M(v)_{s,\pi(s)} = [\T^\M_{\pi}(v)]_s - v_s$ by definition. 
Additionally, since $\vstar^\M$ is the unique fixpoint of $\T_*^\M$ (\cref{lemma:value-bound}) and $\T_*^\M$ maximizes over all actions, $\slack^\M(v^\M_*)_{s,a} \leq 0$, i.e., all actions have non-positive advantage at $v^\M_*$. 
Moreover, since $v_{\pi}^\M$ is the unique fixpoint of $\T_{\pi}^\M$, the policy $\pi$ is optimal if and only if every action has zero advantage with respect to $\vstar^\M$, i.e., $\slack^\M(\vstar^\M)_{s,\pi(s)} = 0$ for all $s \in \S$.

\paragraph{Discarding Actions.}
For $\gamma$-DMDP $\M = (\S, \actionset, p, r)$ and $R \subseteq \actionset$, we write $\M \setminus R$ as removing the actions $R$ from $\M$.
More precisely, it is the $\gamma$-DMDP $\M^\prime = (\S, \actionset \setminus R, p^\prime, r^\prime)$ such that $p^\prime$ and $r^\prime$ are $p$ and $r$ restricted to $\actionset\setminus R$, respectively.

\paragraph{Notational Simplification.} Throughout the paper, when the $\gamma$-DMDP $\M$ is clear from context, we may omit $\M$ in the superscript in our notation (e.g., we may write $v_*$ instead of $v_*^\M$).

\section{From Approximate Values to Exact Policies} \label{sec:reduction}

In this section, we present our general framework (\Cref{alg:reduction}) for solving $\gamma$-DMDPs by reducing the problem to policy evaluations and computing approximate values (as mentioned in \Cref{sec:results}). We first analyze the framework and then instantiate it to prove \cref{thm:det-reduction}. 

At a high level, our framework for solving $\gamma$-DMDPs iteratively discards \emph{suboptimal actions}, i.e., actions not chosen by any optimal policy. This is a common technique for obtaining strongly polynomial time algorithms. For example, in linear programming, irrelevant constraints are discarded (see, e.g.,~\cite{DadushNV20}), and
in submodular function minimization, 
elements that are deduced to not be in any minimizers are discarded (see, e.g., \cite{DVZ21}).
For $\gamma$-DMDPs, the IPM of \cite{Ye05} discards suboptimal actions, and  \cite{Ye11a} argues that iterations of policy iteration will no longer consider certain suboptimal actions after a few iterations. 

To discard suboptimal actions, our framework starts from a policy $\pi$, computes $\epsilon$-optimal values $v$, for a suitable value of $\epsilon$, and then considers the advantage function $\slack(v)_{s, a} = (r_{s,a} + \gamma \cdot \langle p(s,a), v\rangle) - v_s$ with respect to $v$ to detect such suboptimal actions. 
In particular, using known arguments for reasoning about $\gamma$-DMDPs we can show that if $\slack(v)_{s, a}$ is sufficiently negative, i.e., less than $- \epsilon (1+\gamma)$, then no optimal policy $\pi_*$ contains $(s, a)$. 
To reason about what value of $\epsilon$ suffices to guarantee that at least one $(s, a)$ satisfies $\slack(v)_{s, a} < - \epsilon (1+\gamma)$ and can thus be discarded, we follow an 
analysis similar to \cite{Scherrer13} and note that for any policy $\pi$, 
we can relate the most negative $\slack(\vstar)_{s, \pi(s)}$ to the
suboptimality $\|v_{\pi} - \vstar\|_{\infty}$ of $\pi$ (see, e.g., \Cref{lemma:value-bound}). 
Consequently, by relating $\slack(\vstar)_{s, \pi(s)}$ to $\slack(v)_{s, \pi(s)}$, we observe that it suffices to compute $\epsilon$-optimal values $v$ for $\epsilon$ on the order of the suboptimality of $\pi$ up to $\poly(1-\gamma)$ factors (see \Cref{lemma:discard}). 

To compute $\epsilon$-optimal values $v$ for $\epsilon \le \|v_{\pi} - \vstar\|_{\infty} \cdot \poly(1-\gamma)$, we develop and apply an algorithm \texttt{MultValApprox},
 which can be thought of as an algorithm for computing approximate values to multiplicative accuracy relative to the suboptimality of $\pi$. 
\texttt{MultValApprox} first computes $v_{\pi}$. 
It is known that evaluating the values $v_{\pi}$ can be done in $O(\Stot^{\omega})$ by solving a linear system.
\begin{fact}[Policy Evaluation in Fast Matrix Multiplication Time] There is an algorithm that, given any policy $\pi$ in $\gamma$-DMDP $\M = (\S, \actionset, p, r)$, outputs $v^{\M}_{\pi}$ in $O(\Stot^{\omega})$ time.
  \label[fact]{fact:compute-value}
\end{fact}

\begin{proof}
 It is known that $v_{\pi}^\M = (\mathbf{I} - \gamma \P_{\pi})^{-1} r_{\pi}$, where $\P_{\pi} \in \R_{\geq 0}^{\S \times \S}$ is such that $[\P_{\pi}]_{ij} \defeq p(i,\pi(i))_j$ for all $i, j \in \S$ and $r_{\pi} \in \R^\S$ is such that $[r_{\pi}]_i \defeq r_{i,\pi(i)}$ for all $i \in \S$ (see, e.g., \cite{Puterman94}). 
 Thus, $v_{\pi}^{\M}$ can be computed by solving a linear system which takes $O(\Stot^\omega)$ time.
\end{proof}
After the evaluation, \texttt{MultValApprox} discards actions with very negative advantages, meaning advantages more negative than $- \slack_{\max}^{\M}(v_{\pi}) \cdot (1+\gamma)$, where 
\begin{equation}
\label{def:Delta_max}
    \slack_{\max}^{\M}(v_{\pi}) \defeq \max_{(s,a)\in \actionset} \slack^\M(v_{\pi})_{s,a} 
\end{equation}
as they are not chosen by any optimal policy (see \Cref{claim:max-min-vals}).
Next, it shifts rewards using $v_\pi$ by considering a new $\gamma$-DMDP $\M'$ where the rewards of remaining actions are replaced by the advantages $\slack^{\M}(v_{\pi})$. Finally, it 
applies a
$\delta_{\gamma}$-approximate $\gamma$-DMDP algorithm $\mathsf{ApxALG}$ for $\delta_{\gamma} = \poly(1-\gamma)$ to $\M'$, which yields $\epsilon$-optimal values $v$ for $\M$ for explicit $\eps = O(\Delta_{\max}^\M(v_{\pi}) \cdot \poly(1-\gamma)) = O(\|v_{\pi} - \vstar\|_{\infty} \cdot \poly(1-\gamma))$ (see \Cref{lemma:solve-to-scale}). 

\begin{algorithm}[htp!]
\caption{Solving $\gamma$-DMDPs Exactly Using an Approximate $\gamma$-DMDP Solver}
\label{alg:reduction}
\SetKwProg{Fn}{Function}{:}{}
\SetKwFunction{apxval}{MultValApprox}
\SetKwFunction{solve}{SolveDMDP}
\SetKwRepeat{Do}{do}{while}

\SetEndCharOfAlgoLine{}

\Fn{\solve{$\M$}} {
\For{$t = 0, 1, \ldots$ \label{line:main_for}}{
  \tcp{In the randomized reductions considered in \Cref{sec:rand_reduction}, we sample a uniformly random $\pi$.}
  Select a policy $\pi$ of $\M$.\;\label{line:policy}
  \BlankLine
  \tcp{Compute $\eps$-optimal values for $\eps \leq \frac{1-\gamma}{3(1+\gamma)} \cdot \|v_{\pi} - \vstar\|_{\infty}$ (see \cref{sec:shifting}).}
  $(v, \eps) \gets$ \apxval{$\M, \pi$}. \label{line:aprox_val_solve}\;
  \BlankLine
  \tcp{Use $v$ to detect and discard suboptimal actions, $D$.}
  $\M \gets \M \setminus D$  where $D = \{(s,a)\in \actionset: \slack^\M(v)_{s,a} < - \epsilon (1+\gamma)\}\}$. \label{line:discard_actions} \tcp*{discard $D$}

  \BlankLine
  \tcp{If no actions are discarded, the policy must be optimal.}
  \lIf(\tcp*[f]{For analysis define $T \defeq t$ on this line}){$D = \emptyset$}{\Return $\pi$}
}
 }

\;

\Fn{\apxval{$\M, \pi$}}{

  $X^{\pi} \gets \{(s,a)\in \actionset: \slack^\M(v_{\pi})_{s,a} < -\slack_{\max}^{\M}(v_{\pi}) \cdot (1+\gamma)\}$. \tcp*{actions to ignore}
  Compute $v^\prime\gets\mathsf{ApxALG}(\M[\slack^\M(v_{\pi})]\setminus X^{\pi}, \delta_\gamma)$ where $\delta_\gamma \defeq \frac{1-\gamma}{3(1+\gamma)^2}$, where\;
  $\mathsf{ApxALG}$ is a $\delta_{\gamma}$-approximate $\gamma$-DMDP solver (see \cref{def:apx_values}).
  
  \Return{$(v, \eps)$ where $v \gets v^\prime + v^\M_{\pi}$ and $\eps \gets \slack_{\max}^{\M}(v_{\pi}) \cdot \frac{1-\gamma}{3(1+\gamma)}$.
  }
}
\end{algorithm}

In the remainder of this section we break down our analysis of \cref{alg:reduction} (and the use of it to prove \cref{thm:det-reduction}) 
into multiple parts. 
First, in \cref{sec:discard}, we prove \cref{lemma:discard}, which shows that solving the $\gamma$-DMDP to an accuracy proportional to the suboptimality of a policy $\pi$, meaning $\|v_* - v_{\pi}\|_{\infty}$, suffices to detect actions that have a negative enough advantage and can thus be discarded (as they are provably not part of any optimal solution). 
In \cref{sec:shifting} we then analyze \texttt{MultValApprox}. 
Finally, in \cref{sec:combine}, we leverage the two parts to prove the correctness of our reduction framework and bound the running time of each iteration (see \cref{thm:framework_correctness}). Finally, \cref{thm:det-reduction} follows from this analysis straightforwardly.

\subsection{Detecting Suboptimal Actions} \label{sec:discard}

In this section, we state and prove \cref{lemma:discard} which provides a bound on the accuracy to which it suffices to solve for the approximately optimal values to detect suboptimal actions. Similar bounds and analysis appeared in \cite{Scherrer13}, and we provide a self-contained proof for completeness and to enable our specific algorithms. 

Our analysis, \cite{Scherrer13}, and \cref{lemma:discard} all connect the suboptimality of a policy $\pi$ to a key quantity $\slack_{*, \min}^{\M}(\pi)$ defined as the minimum advantage value of the optimum values for a state-action pair in $\pi$, i.e., 
\begin{equation}\label{eq:z_pi}
    \slack_{*, \min}^{\M}(\pi) \defeq \min_{s\in \S}\slack^\M(v^\M_*)_{s,\pi(s)}
\end{equation}
Although we do not compute $\slack_{*, \min}^{\M}(\pi)$ exactly in \cref{alg:reduction}, we use it in our analysis to bound the progress made by our algorithm in discarding actions. 
In particular, we leverage in the analysis of our framework (\cref{alg:reduction}) that when $v, \pi$ satisfy that
$\|v - v^\M_*\|_{\infty}$ is small enough relative to $\|v^\M_{\pi} - v^\M_*\|_{\infty}$, any
$(s,a) \in \actionset$ with $\slack^\M(v^\M_*)_{s, a} \leq \slack_{*, \min}^{\M}(\pi)$ 
will be discarded as it does not belong to 
the set of optimal actions. Additionally, \Cref{lemma:discard} implies that solving up to accuracy $\eps < \|v^\M_{\pi} - v^\M_*\|_\infty \cdot \frac{1-\gamma}{2(1+\gamma)}$ suffices to be able to discard at least one action this way.

\begin{lemma}[\cite{Scherrer13}]
  \label[lemma]{lemma:discard}
  Let $v \in \R^\S$ be $\eps$-optimal for $\gamma$-DMDP $\M = (\S, \actionset, p, r)$. 
  If $\slack^\M(v)_{s,a} \allowbreak < -\eps (1+\gamma)$ for some $(s, a) \in \actionset$, 
  no optimal policy $\pi_*$ of $\M$ has $\pi_*(s) = a$.
  Moreover, if $\eps < \|v^\M_{\pi} - v^\M_*\|_\infty \cdot \frac{1-\gamma}{2(1+\gamma)}$, then any 
  $(s,a)\in\actionset$ with $\slack^\M(v^\M_*)_{s,a} \leq \slack_{*, \min}^{\M}(\pi)$ will satisfy $\slack^\M(v)_{s,a} < -\eps (1+\gamma)$. 
\end{lemma}

To prove \Cref{lemma:discard}, 
we provide two helper lemmas, \Cref{lemma:slack-diff,lemma:large-slack-opt}.

\begin{restatable}{lemma}{SlackDiff}  \label[lemma]{lemma:slack-diff}
  For a $\gamma$-DMDP $\M = (\S, \actionset, p, r)$ and any $u, v \in \R^\S$ and  $(s, a) \in \actionset$,
  \[\left|\slack^\M(u)_{s, a} - \slack^\M(v)_{s,a}\right| \leq (1 + \gamma) \|u - v\|_{\infty}.\] 
\end{restatable}

\begin{proof}
  It follows from the definition of $\slack^\M(u)_{s,a}$ 
  that 
  \begin{align*}
    \left|\slack(u)_{s,a} - \slack(v)_{s,a}\right|
    &=
    \left|\left(r_{s,a} + \gamma \cdot \langle p(s, a), u \rangle - u_s \right) - \left(r_{s,a} + \gamma \cdot \langle p(s,a), v \rangle - v_s \right)\right| \\
    &= |\gamma \cdot \langle p(s, a), u - v \rangle - u_s + v_s|
    \leq \gamma \cdot |\langle p(s, a), u - v\rangle| + |u_s - v_s| \\
    &\leq (1 + \gamma)\|u - v\|_{\infty},
  \end{align*}
  where in the last inequality we used that $\langle p(s, a), u - v \rangle \le \|p(s, a)\|_1 \cdot \|u - v\|_{\infty} = \|u - v\|_{\infty}$ and that $|u_s - v_s| \leq \|u - v\|_{\infty}$.
\end{proof}

\begin{restatable}{mylemma}{LargeSlackOpt}  \label[lemma]{lemma:large-slack-opt}
    For any policy $\pi$ in $\gamma$-DMDP $\M = (\S, \actionset, p, r)$, $\slack_{*, \min}^{\M}(\pi) \leq -(1-\gamma)\|v_\pi^{\M} - v_*^{\M}\|_\infty$. 
\end{restatable}

\begin{proof}
    Note that $\slack(\vstar)_{s, \pi(s)} =  r_{s, \pi(s)} + \gamma \cdot \langle p(s, \pi(s)), \vstar \rangle - [\vstar]_s = [\T_{\pi}(\vstar)]_s - [\vstar]_s$. Consequently, using \eqref{eq:gen_contract}, which, as discussed in \Cref{sec:prelim}, follows from \cref{lemma:value-bound}, we obtain 
    \[ \|v_{\pi} - \vstar\|_{\infty} \leq \frac{1}{1-\gamma} \|\vstar - \T_{\pi}(\vstar)\|_\infty = \frac{1}{1-\gamma} \max_{s \in S} |\slack(\vstar)_{s, \pi(s)}|\,.
    \]
    Rearranging and using that $\max_{s\in S}|\slack(\vstar)_{s,\pi(s)}| = -\slack_{*, \min}(\pi)$ proves the lemma.
\end{proof}

We are now ready to prove \cref{lemma:discard}.

\begin{proof}
    By \cref{lemma:slack-diff}, for any $(s, a) \in \actionset$, we have
    \[|\slack(v_*)_{s, a} - \slack(v)_{s, a}| \leq (1 + \gamma) \|\vstar - v\|_{\infty} \le (1 + \gamma) \eps.\]
    Hence, $\slack(v)_{s, a} < -\eps (1+\gamma)$ implies $\slack(v_*)_{s, a} < 0$, and thus no optimal policy $\pi$ of $\M$ has $\pi(s) = a$. 
    Next, by \cref{lemma:large-slack-opt}, any action $(s, a)$ with $\slack(v_*)_{s, a} \leq \slack_{*, \min}(\pi)$ has $\slack(v_*)_{s, a} \leq -\eta (1-\gamma)$.
    This in turns by \cref{lemma:slack-diff} shows that $\slack(v)_{s, a} \leq -\eta (1-\gamma) + (1+\gamma) \eps$. 
    For $\eps < \eta \cdot \frac{1-\gamma}{2(1+\gamma)}$, we have $-\eta  (1-\gamma) + (1+\gamma)\eps < -\eps (1+\gamma)$ which concludes the proof. 
\end{proof}

\subsection{Approximate Values via Reward Shifting} \label{sec:shifting}

In this section, we state and prove \cref{lemma:solve-to-scale}, which 
provides the correctness and efficiency guarantee of  \texttt{MultValApprox}, whose pseudocode was given in \cref{alg:reduction}. 
As discussed, this subroutine is applied to obtain approximate values $v$ of a $\gamma$-DMDP to the accuracy required by \cref{lemma:discard}. 

\begin{restatable}{lemma}{SolveToScale}
  Given a policy $\pi$ of $\gamma$-DMDP $\M = (\S, \actionset, p, r)$, 
  the \emph{\texttt{MultValApprox}} subroutine in \cref{alg:reduction} computes $\eps \leq \|v_{\pi}^\M - \vstar^\M\|_{\infty} \cdot \frac{1-\gamma}{3(1+\gamma)}$ and $\eps$-optimal values $v$ and can be implemented in deterministic 
  $O(\Stot^{\omega} + \Stot\Atot) + \T_{\Det}(\Stot, \Atot, \allowbreak \Theta(1-\gamma))$ time and can be implemented in randomized $O(\Stot^{\omega} + \Stot\Atot) + \T_{\Rnd}(\Stot, \Atot, \Theta(1-\gamma))$ expected time. 
  \label[lemma]{lemma:solve-to-scale}
\end{restatable}

To motivate the \texttt{MultValApprox} subroutine, note that it is not immediately clear how to efficiently run the approximate $\gamma$-DMDP solvers on the original $\gamma$-DMDP instance to accuracy $\|v_{\pi} - v_*\|_{\infty} \cdot \poly(1-\gamma)$ to enable applying \cref{lemma:discard}. 
This is because the ratio $\frac{\|r\|_{\infty}}{\|v_{\pi} - \vstar\|_{\infty}}$ can be very large, and the running time of the approximate solver depends on this ratio.

To bypass this obstacle, the \texttt{MultValApprox} subroutine appropriately modifies the rewards in the $\gamma$-DMDP instance that the solvers run on (as aforementioned). 
In particular, the first step of \texttt{MultValApprox} 
is to shift the reward so that the actions taken by the policy $\pi$ become the ``baseline'', i.e., have value $0$.
This is done by using the advantage function $\slack(v_{\pi})$ with respect to $v_{\pi}$ as the reward instead.
Note that given $\pi$, $\slack(v_{\pi})_{s,a}$ can be computed in $O(\Stot^{\omega} + \Stot\Atot)$ time for all $(s, a) \in \actionset$ using \Cref{fact:compute-value}. 
The shifted MDP is formally defined as follows.

\begin{definition}[Shifted MDP]
  For a $\gamma$-DMDP $\M = (\S, \A, p, r)$ and a policy $\pi$ of $\M$, $M^\pi \defeq (\S, \A, p, \slack^\M(v_{\pi}^\M))$ is the $\gamma$-DMDP that has the same states, actions, and transition probabilities as $\M$, but with reward vector $\slack^\M(v_{\pi}^\M)$ instead of $r$.
\end{definition}
To prove \Cref{lemma:solve-to-scale}, we provide a helper lemma and two claims.
We first show that, for any policy $\pi^\prime$, the value vector of $\pi^\prime$ in $\M^\pi$ is merely the value vector of $\pi^\prime$ in $\M$ shifted by $v_{\pi}$. 

\begin{lemma}
    For any policies $\pi$ and $\pi'$ of $\gamma$-DMDP $\M = (\S, \actionset, p, r)$, $v^{\M}_{\pi^\prime} = v^{\M^\pi}_{\pi^\prime}+v^{\M}_{\pi}$, and therefore $v^\M_* = v^{\M^\pi}_* + v^\M_{\pi}$.
  \label[lemma]{lemma:value-shift}
\end{lemma}

\begin{proof}
  Since $v^{\M^\pi}_{\pi^\prime}$ is the unique fixed point of the Bellman on-policy operator $\T^{\M^\pi}_{\pi^\prime}$, it satisfies
  \begin{align*}
    [v^{\M^\pi}_{\pi^\prime}]_s
    &= \slack^\M(v_{\pi})_{s,\pi^\prime(s)} + \gamma \cdot \langle p(s, \pi^\prime(s)), v^{\M^\pi}_{\pi^\prime} \rangle\\
    &= \left(r_{s,\pi^\prime(s)} - [v^\M_{\pi}]_s + \gamma \cdot p(s, \pi^\prime(s)), v^\M_{\pi} \rangle\right) + \gamma \cdot \langle p(s, \pi^\prime(s)), v^{\M^\pi}_{\pi^\prime} \rangle,
  \end{align*}
  which by rearranging becomes
  \begin{align*}
    [v^{\M^\pi}_{\pi^\prime}]_s + [v^{\M}_{\pi}]_s
    &= r_{s,\pi^\prime(s)} + \gamma \cdot \langle p(s,\pi^\prime(s)),  v^{\M^\pi}_{\pi^\prime} + v^{\M}_{\pi} \rangle,
  \end{align*}
  and therefore $v^{\M}_{\pi^\prime} = v^{\M^\pi}_{\pi^\prime}+v^\M_{\pi}$.
\end{proof}

This shows that the $\M^\pi$ has the same optimal policy, and thus we can focus on identifying suboptimal actions in $\M^\pi$.
Observe that $\slack^\M(v_{\pi})_{s, \pi(s)} = 0$ for all states $s\in \S$.
That is, each state has at least one action with reward $0$.
Recall the definition of $\slack_{\max}^{\M}(v_{\pi})$ in \eqref{def:Delta_max} and note that $\slack_{\max}^{\M}(v_{\pi}) \ge 0$. 
Next, we provide \Cref{claim:bound} and \Cref{claim:max-min-vals}.
\begin{restatable}{claim}{Bound}
   For every policy $\pi$ of $\gamma$-DMDP $\M = (\S, \actionset, p, r)$, 
   $v_{\pi}^{\M^\pi} = \vec{0}$ and $\|v_{*}^{\M^\pi}\|_{\infty} \geq \slack_{\max}^{\M}(v_{\pi})$. 
  \label[claim]{claim:bound}
\end{restatable}

\begin{proof}
  That $\|v_{\pi}^{\M^\pi}\|_{\infty} = 0$ follows by \Cref{lemma:value-shift} since each $(s, \pi(s))$ satisfies $\slack^\M(v_{\pi})_{s,\pi(s)} = 0$. 
  For the bound on $\|\vstar^{\M^\pi}\|_{\infty}$, consider the action $(s_*,a_*) \in \actionset$ that maximizes $\slack^\M(v_{\pi})_{s_*,a_*}$ which by definition has $\slack^\M(v_{\pi})_{s_*,a_*} = \slack_{\max}^{\M}(v_{\pi})$.
  Let $\pi^\prime$ be the policy that is the same as $\pi$ except at state $s_*$ where $\pi^\prime(s_*) = a_*$. 
  Note that since for every $s$, the reward of $(s, \pi(s))$ is $0$ in $\M^{\pi}$, it
  follows easily that $[v_{\pi^\prime}^{\M^\pi}]_{s_*} \geq \slack_{\max}^{\M}(v_{\pi})$, as the rewards from following $\pi^\prime$ are all non-negative. 
\end{proof}

\Cref{claim:bound} implies that actions with sufficiently negative rewards do not belong to any optimal policy and can thus be ignored.
This allows us to bound the range of the rewards that we need to consider for the approximate value optimization.

\begin{restatable}{claim}{MaxMinVals}
  For any policy $\pi$ of $\gamma$-DMDP $\M = (\S, \actionset, p, r)$,
  no optimal policy $\pi_*$ of $\M$  has $\pi_*(s) = a$ for $\slack^\M(v_{\pi})_{s,a} < -\slack_{\max}^{\M}(v_{\pi}) \cdot (1+\gamma)$.
  
  \label[claim]{claim:max-min-vals}
\end{restatable}

\begin{proof}
  By \cref{claim:bound,lemma:value-shift}, we have $\|v_{\pi} - v_*\| \geq \slack_{\max}(v_{\pi})$.
  Thus, \cref{lemma:discard} asserts that $(s, a)$ with $\slack(v_{\pi})_{s,a} < -\slack_{\max}^{\M}(v_{\pi}) \cdot (1+\gamma)$ is not part of any optimal policy. 
\end{proof}

Based on \cref{claim:max-min-vals}, we let
$X^{\pi} \defeq \{(s, a): \slack^\M(v_{\pi})_{s,a} < -\slack_{\max}^{\M}(v_{\pi}) \cdot (1+\gamma)\}$ be the set of actions that have very negative advantages in $\M$ and can be
ignored.
We now prove \cref{lemma:solve-to-scale}.

\begin{proof}[Proof of \cref{lemma:solve-to-scale}]
  We first bound the running time of \texttt{MultValApprox}. 
  \texttt{MultValApprox} first computes $X^{\pi}$ which can be done in $O(\Stot^\omega + \Stot\Atot)$ time by calculating $v_{\pi}$ and  $\slack^\M(v_{\pi})_{s,a}$ for all $(s, a) \in \A$ (which is also used in 
  computing the return value of $\eps \defeq \slack_{\max}^{\M}(v_{\pi}) \cdot \frac{1-\gamma}{3(1+\gamma)}$).
  Since $\frac{1-\gamma}{3(1+\gamma)^2} = \Theta(1-\gamma)$, 
  the call to a $\delta_\gamma$-approximate solver can be implemented in 
  $\T_{\Det}(\Stot,\Atot, \Theta(1-\gamma))$ deterministic time or $\T_{\Rnd}(\Stot,\Atot, \Theta(1-\gamma))$ randomized expected time, which establishes the running time of \texttt{MultValApprox}.

We now prove the correctness of \texttt{MultValApprox}.
Let $\M^\prime \defeq \M^\pi \setminus X^{\pi}$, i.e., the $\gamma$-DMDP that only contains actions in $\actionset\setminus X^{\pi}$ and has the reward vector $\slack^\M(v_{\pi})$. 
Note that since $\delta_{\gamma} \defeq \frac{1-\gamma}{3(1+\gamma)^2} = \frac{\eps}{\slack_{\max}^{\M}(v_{\pi})\cdot(1+\gamma)}$ and 
the maximum magnitude of the reward in $\M^\prime$ is bounded by $\slack_{\max}^{\M}(v_{\pi})\cdot(1+\gamma)$ by definition of $X^{\pi}$, 
the value vector $v^\prime$ it outputs satisfies $\|v^\prime - \vstar^{\M^\prime}\|_{\infty} \leq \eps$ (see \cref{def:apx-algo} for the definition of an $\delta$-approximate algorithm). Thus, 
it follows from \cref{lemma:value-shift,claim:max-min-vals} that $v$ is $\eps$-optimal for the original $\gamma$-DMDP $\M$. 
Lastly, 
note that by \cref{claim:bound,lemma:value-shift} we know that $\|\vstar^\M - v_{\pi}^\M\|_{\infty} \geq \slack_{\max}^{\M}(v_{\pi})$, which implies $\eps \leq \|v_{\pi}^\M - \vstar^\M\|_{\infty} \cdot \frac{1-\gamma}{3(1+\gamma)}$, so $\eps$ satisfies our requirement. 
This concludes the proof.
\end{proof}

\subsection{Putting Everything Together} \label{sec:combine}

Here we use \cref{lemma:discard,lemma:solve-to-scale} to prove the correctness and efficiency of \cref{alg:reduction}, our reduction from exact $\gamma$-DMDP solvers to approximate $\gamma$-DMDP solvers. This is captured by \Cref{thm:framework_correctness} below. 
We then use \Cref{thm:framework_correctness} to prove \Cref{thm:det-reduction}.

\begin{theorem}[Correctness and Efficiency of \cref{alg:reduction}]
\label{thm:framework_correctness}
    Given a $\gamma$-DMDP $\M = (\S, \actionset, p, r)$, 
    \Cref{alg:reduction} computes an optimal policy $\pi$ of $\M$ and
    each iteration of the for loop can be implemented 
    in $O(\Stot^{\omega} + \Stot\Atot) + \T_{\Det}(\Stot, \Atot, \Theta(1-\gamma))$ deterministic time and $O(\Stot^{\omega} + \Stot\Atot) + \T_{\Rnd}(\Stot, \Atot, \Theta(1-\gamma))$ randomized expected time. 
\end{theorem}

\begin{proof}
    First, we prove the correctness of \cref{alg:reduction}. 
    For this, first note that by \cref{lemma:discard}, every $(s, a) \in D$, for $D$ defined in \Cref{line:discard_actions}, satisfies that $\pi(s) \ne a$ for all optimal policies $\pi$ of $\M$ since $v$ is $\eps$-optimal by \cref{lemma:solve-to-scale}.
    Next, we show that if $\slack_{*, \min}^{\M}(\pi) < 0$ in some iteration of the for loop in \Cref{line:main_for},  $D \supseteq D^\M(\pi)$, for $D^\M(\pi) \defeq \{(s,a)\in \actionset: \slack^{\M}(v_*)_{s, a}\leq \slack_{*, \min}^{\M}(\pi)\}$. 
    For this, note that if $\slack_{*, \min}^{\M}(\pi) < 0$, then $\|v_{\pi} - \vstar\|_{\infty} > 0$, and thus $\eps < \|v_{\pi} - \vstar\|_{\infty} \cdot \frac{1-\gamma}{2(1+\gamma)}$.
  By \cref{lemma:discard}, any $(s,a)$ with $\slack^{\M}(v_*)_{s, a} \leq \slack_{*, \min}^{\M}(\pi)$ must have $\slack^{\M}(v)_{s, a} < -\eps  (1+\gamma)$.
  Therefore, in this case we must have $D\supseteq D^\M(\pi) = \{(s,a)\in \actionset: \slack^{\M}(v_*)_{s, a}\leq \slack_{*, \min}^{\M}(\pi)\}$, and so $|D| \ge 1$. 
    Consequently, if $D = \emptyset$ in some iteration of the for loop in \Cref{line:main_for}, we have $\slack_{*, \min}(\pi) = 0$, since otherwise the $(s,\pi(s))$ with $\slack(\vstar)_{s,\pi(s)} = \slack_{*, \min}(\pi)$ will be in $D$. 
  This shows that whenever \cref{alg:reduction} terminates, the policy $\pi$ must have $\slack_{*, \min}(\pi) = 0$, i.e., is optimal.

  Now, we bound the running time of \cref{alg:reduction}. To do so, we bound the running time of each iteration of the for loop in \Cref{line:main_for}. 
  Note that selecting a policy $\pi$ takes $O(\Atot)$ time.
  $O(\Stot\Atot)$ time, respectively. 
  Computing approximate values, by \cref{lemma:solve-to-scale}, takes $O(\Stot^\omega + \Stot\Atot) + \T_{\Det}(\Stot, \Atot,\Theta(1-\gamma))$ time or $O(\Stot^\omega + \Stot\Atot) + \T_{\Rnd}(\Stot, \Atot,\Theta(1-\gamma))$ randomized expected time. 
  Computing $D$ in \Cref{line:discard_actions} takes $O(\Stot \Atot)$ time by computing $\slack(v)_{s,a}$ for all $(s, a) \in \A$. 
  Thus, overall, the (expected) running time per iteration is bounded by $O(\Stot^\omega + \Stot\Atot) + \T_{\Det}(\Stot, \Atot,\Theta(1-\gamma))$ or $O(\Stot^\omega + \Stot\Atot) + \T_{\Rnd}(\Stot, \Atot,\Theta(1-\gamma))$. 
\end{proof}

\begin{remark}
\label{rem:strongly_poly}
    It is easy to see that our framework itself runs in strongly polynomial time even in the stricter Turing model (i.e., the space complexity of our framework is polynomial of the input bit complexity).
    In particular, our framework first evaluates $v_{\pi}$, which can be implemented in $O(\Stot^\omega)$ strongly polynomial time in the Turing model (see, e.g.,~\cite{Storjohann05}). It then only performs a constant number of basic operations (e.g., addition, subtraction, multiplication) involving $v_{\pi}$ and the inputs.
    This includes multiplication with $1/(1-\gamma)$ which has the same bit complexity as the input value $\gamma$.
    
    However, since the framework invokes a black-box approximate solver (e.g., the recent linear programming algorithms), we do not formally claim our end-to-end exact algorithms run in strongly polynomial time in the Turing model, as bounding the bit complexity of the operations performed by these approximate solvers is beyond the scope of this paper.
    That being said, we remark that there has been some recent work~\cite{GhadiriPV23} that does this for a class of continuous optimization methods and some linear programming algorithms. 
\end{remark}

We now prove our theorem regarding the deterministic reduction, which we restate below. 

\DetReduction*

\begin{proof}
  We prove that \Cref{alg:reduction} respects the desired properties. 
  Note that its correctness and running time follow from \Cref{thm:framework_correctness} provided that the number of iterations $T$ of
  \cref{alg:reduction} satisfies $T = O(\Atot)$. However, $T = O(\Atot)$ as in every iteration of the for loop in which the algorithm does not terminate, the set $D$ contains at least one action, and therefore the number of actions in $\M$ decreases by at least $1$.
\end{proof}

\section{Randomized Reduction}\label{sec:rand_reduction}

In this section, we 
instantiate our framework \Cref{alg:reduction} to obtain our randomized reduction, with correctness and efficiency guarantees stated in \Cref{thm:rand-reduction}. 
In particular, instead of choosing an arbitrary policy $\pi$ in \Cref{line:policy} of \cref{alg:reduction}, we select a policy uniformly at random from the remaining policies (i.e., policies which do not choose actions already discarded), i.e., each remaining policy has an equal probability of being selected.  
This approach is motivated by the fact that the worst-case running time for the algorithm we give for \Cref{thm:det-reduction} 
is achieved when only one action is discarded at a time, which intuitively corresponds to picking the ``most suboptimal'' policy $\pi$ in each iteration. We show that picking a random policy instead decreases the number of iterations of \cref{alg:reduction} in expectation.

To analyze the iteration count when randomly selecting $\pi$ in \cref{line:policy}, we consider the number of $(s, a) \in \actionset$ with $\slack^{\M}(\vstar^\M)_{s, a} \leq \slack_{*, \min}^{\M}(\pi) = \min_{s'} \slack^\M(\vstar^\M)_{s', \pi(s')}$ (i.e., all actions with more negative advantages than ``worst'' action chosen by $\pi$) for policy $\pi$ picked uniformly at random. 
It is easy to see that obtaining $\epsilon$-optimal values $v$ for the value of $\epsilon$ chosen in \Cref{alg:reduction} allows us, by \Cref{lemma:discard}, to discard all aforementioned actions.
To reason about how many such actions we can discard, we let $\Pi(\M)$ be the set of all policies of $\M$, define a potential $\Phi(\M) \defeq |\Pi(\M)| = \prod_{s}|\A_s|$ to be the number of policies, and let $\pi \sim \Pi(\M)$ denote sampling $\pi$ uniformly at random from $\Pi(\M)$.
In \cref{lemma:potential}, we show that, in expectation for $\pi \sim \Pi(\M)$,
$\Phi(\M)$ drops by a factor of two after discarding actions more suboptimal than $\slack^\M_{*,\min}(\pi)$. Consequently, this allows us to reduce the number of calls to our procedure from $\Atot - \Stot$ to $O(\Stot \log \Atot)$ in expectation.

\begin{restatable}{mylemma}{Potential}
  Let $\M = (\S, \actionset, p, r)$ be a $\gamma$-DMDP where $\Phi(\M) > 0$.
  Then, letting $D^\M(\pi) \defeq \{(s, a): \slack^{\M}(\vstar^\M)_{s, a} \leq \slack_{*, \min}^{\M}(\pi)\}$ for policy $\pi$ we have $\mathbb{E}_{\pi \sim \policy(\M)}[\Phi(\M\setminus D^\M(\pi))] \leq \frac{1}{2} \Phi(\M)$.
  \label[lemma]{lemma:potential}
\end{restatable}

\begin{proof}
  We proceed by an induction on $\Phi(\M)$.
  For the base case where $\Phi(\M) = 1$, there exists exactly one policy $\pi$ of $\M$.
  This means that all actions are optimal, i.e., have advantage zero, and thus $D^\M(\pi)$ contains all of them and thus $\Phi(\M \setminus D^\M(\pi)) = 0 \leq \frac{1}{2}$.

  When $\Phi(\M) > 1$, consider an action $(s_*, a_*)$ that minimizes $\slack^\M(\vstar^\M)_{s_*, a_*}$.
  If $\slack^\M(\vstar^\M)_{s_*, a_*} = 0$, then again all actions are optimal and therefore $\Phi(\M \setminus D^{\M}(\pi)) = 0 \leq \frac{1}{2}\Phi(\M)$.
  We thus assume for the remainder of the proof that $\slack^\M(\vstar^\M)_{s_*, a_*} < 0$.
  Let $\M^\prime \defeq \M \setminus \{(s_*,a_*)\}$ be the $\gamma$-DMDP obtained by removing $(s_*, a_*)$ from $\M$. 
  First observe that $(s_*,a_*)\in D^\M(\pi)$ regardless of the choice of $\pi$ and that $\Phi(\M^\prime) = \Phi(\M) \cdot \frac{|\A_{s_*}|-1}{|\A_{s_*}|}$ by definition of $\Phi$.
  Consider the process of selecting a uniformly random policy $\pi$ of $\M$.
  We consider the two cases of whether $\pi(s_*) = a_*$ or not and bound $\mathbb{E}[\Phi(\M\setminus D^\M(\pi))]$ by
  \begin{equation}
  \begin{split}
    \mathbb{E}[\Phi(\M\setminus D^\M(\pi))]
    &= \Pr[\pi(s_*) \neq a_*] \cdot \mathbb{E}[\Phi(\M\setminus D^\M(\pi)) \mid \pi(s_*) \neq a_*]
    \\ &+ \Pr[\pi(s_*) = a_*] \cdot \mathbb{E}[\Phi(\M\setminus D^\M(\pi)) \mid \pi(s_*) = a_*].
  \end{split}
  \label{eq:condition}
  \end{equation}
  To compute the first term, note that the event $\pi(s_*) \neq a_*$ happens with probability $1 - \frac{1}{|\A_{s_*}|}$.
  The distribution of $\pi \sim \Pi(\M)$ conditioned on $\pi(s_*) \neq a_*$ is the same as the distribution $\pi \sim \Pi(\M^\prime)$.
  Moreover, observe that we have $\vstar^{\M} = \vstar^{\M^\prime}$ since we only remove a suboptimal action, and therefore $\slack^\M(v_*^\M)_{s,a} = \slack^{\M^\prime}(\vstar^{\M^\prime})_{s,a}$ for all $(s,a) \neq (s_*, a_*)$.
  In particular this implies $D^{\M^\prime}(\pi) = D^{\M}(\pi) \setminus \{(s_*,a_*)\}$.
  Because $(s_*,a_*) \in D^\M(\pi)$ regardless of the choice of $\pi$, we have
  \[
    \mathbb{E}_{\pi \sim \policy(\M)}[\Phi(\M \setminus D^\M(\pi)) \mid \pi(s_*) \neq a_*] = 
    \mathbb{E}_{\pi \sim \policy(\M^\prime)}[\Phi(\M^\prime \setminus D^{\M^\prime}(\pi))] \leq \frac{1}{2}\Phi(\M^\prime),
  \]
  where the last inequality follows from the inductive hypothesis.
  For the second term in \eqref{eq:condition}, again since $D^{\M}(\pi)$ contains $(s_*,a_*)$ we have $\Phi(\M\setminus D^\M(\pi)) \leq \Phi(\M^\prime)$.
  Overall, we get
  \begin{align*}
    \mathbb{E}_{\pi \sim \policy(\M)}[\Phi(\M\setminus D^\M(\pi))] &\leq \frac{|\A_{s_*}| - 1}{|\A_{s_*}|} \cdot \frac{1}{2}\Phi(\M^\prime) + \frac{1}{|\A_{s_*}|} \cdot \Phi(\M^\prime)
    = \frac{|\A_{s_*}|^2-1}{2|\A_{s_*}|^2} \cdot \Phi(\M) \\
    & \leq \frac{1}{2}\Phi(\M).
  \end{align*}
\end{proof}

We now use \Cref{lemma:potential} to prove \cref{thm:rand-reduction}, the theorem regarding our randomized reduction, which we restate below.

\RandReduction*

\begin{proof}
  We prove that \Cref{alg:reduction} with uniformly random policy selection respects the desired properties. 
  Note that its correctness and running time follow from \Cref{thm:framework_correctness} provided that the number of iterations $T$ of
  \cref{alg:reduction} satisfies $\mathbb{E}[T] = \tilde{O}(\Stot)$.
  
    To this end, note that by \cref{lemma:discard} the set $D$ contains $D(\pi)$ as long as $\slack_{*, \min}^{\M}(\pi) \neq 0$. 
    Thus, by \cref{lemma:potential}, $\Phi(\M)$ decreases by a factor of $2$ each iteration in expectation, since clearly $\Phi$ is such that $\Phi(\M \setminus D) \leq \Phi(\M \setminus D(\pi))$.
    When $\Phi(\M) \leq 1$, $\M$ has only one policy which is guaranteed to be optimal, and thus we will get $\slack_{*, \min}^{\M}(\pi) = 0$ and $D = \emptyset$ at that point. The value of $\Phi(\M)$ is upper bounded by $\Atot^{\Stot}$, and thus the expected number of iterations $T$ of the for loop until we have an optimal policy is $O(\Stot\log \Atot)$. By \Cref{thm:framework_correctness}, each iteration takes $O(\Stot^{\omega} + \Stot\Atot) + \T_{\Rnd}(\Stot, \Atot, \Theta(1-\gamma))$ randomized expected time, and thus the stated expected runtime of our algorithm follows. 
\end{proof}

\section{A Randomized Policy Iteration Algorithm} \label{appendix:random-policy}

In this section, we 
instantiate our framework (\Cref{alg:reduction}) to obtain a simple randomized variant of policy iteration that alternates between a ``random'' step and $O(\log(1/\gamma)/(1-\gamma))$ policy iteration steps, outperforming the deterministic variant (\cref{alg:randomized-policy-iteration}). 
In particular, we start from a policy selected uniformly at random in \Cref{line:policy}, as in \Cref{sec:rand_reduction}, and,
to implement a $\delta_{\gamma}$-approximate $\gamma$-DMDP solver $\mathsf{ApxALG}$, apply 
the well-known policy iteration algorithm~\cite{howard1960dynamic} that, for each state $s$, simultaneously switches the action taken at $s$ to one that has the best advantage with respect to the current policy $\pi$.
In other words, in each iteration of policy iteration, we update $\pi$ to $\pi^+$ where $\slack^{\mdp}(v_{\pi})_{s, \pi^+(s)} = \max_{a} \slack^{\mdp}(v_{\pi})_{s, a}$ for all $s \in \S$. 

\begin{algorithm}[h]
\caption{A randomized policy iteration algorithm}
\label{alg:randomized-policy-iteration}
\SetKwProg{Fn}{Function}{:}{}
\SetKwFunction{apxval}{MultValApprox}
\SetKwFunction{bestpolicy}{PolicyIterationStep}
\SetKwFunction{randompolicy}{RandomizedPolicyIteration}
\SetKwRepeat{Do}{do}{while}
\SetEndCharOfAlgoLine{}

\Fn{\randompolicy{$\mdp$}} {
  \For{$t = 0, 1, \ldots$} {\label{line:for_loop_rand_PI}
    Select a uniformly random policy $\pi$ of $\M$.\;

    \For{$x = 0, \ldots, \lceil \frac{\ln (3(1+\gamma)/(1-\gamma)^2)}{1-\gamma}\rceil$ \label{line:for_rand_PI}} {
      $\pi \gets$ \bestpolicy{$\mdp, \pi$}.\;
    }
    $\M \gets \M \setminus D$  where $D = \{(s,a)\in \actionset: \slack^\M(v)_{s, a} < - \slack_{\max}^{\M}(v_{\pi}) \cdot \frac{1-\gamma}{3}\}$. \;

  \lIf(\tcp*[f]{For analysis define $T \defeq t$ on this line}){$D = \emptyset$}{\Return $\pi$}
  }
}

\;
\Fn{\bestpolicy{$\mdp, \pi$}}{
    \Return $\pi^+$, where $\pi^+(s) \defeq \argmax_{(s,a)\in \actionset}r_{s,a}+\gamma \cdot \langle p(s,a), v_{\pi} \rangle$.\;\label{line:new-policy}
}
\end{algorithm}

To prove the correctness of \cref{alg:randomized-policy-iteration}, we use the following fact that $\Theta(\log(1/\delta)/(1-\gamma))$  policy iteration steps times suffices to improve the optimality of a policy  $\pi$ by a $\delta < 1$ factor.

\begin{lemma}[{see, e.g.,~\cite[Lemma 2]{Scherrer13}}]
\label[lemma]{lem:one_step_policy_iter}
    For any $\gamma$-DMDP $\M=(\S,\actionset,p,r)$, policy $\pi$ of $\M$, and $\pi^+ \defeq \emph{\texttt{PolicyIterationStep}}(\M, \pi)$ (as in \cref{alg:randomized-policy-iteration}), $\|v^\M_{\pi^+} - v^\M_*\|_{\infty} \le \gamma \|v^\M_{\pi} - v^\M_*\|_{\infty}$.
\end{lemma}

Using \Cref{lem:one_step_policy_iter}, we prove the following corollary. 

\begin{corollary}
  Given a $\gamma$-DMDP $\M = (\S, \actionset, r, p)$, \cref{alg:randomized-policy-iteration} computes an optimal policy $\pi$ of $\M$ in expected $\tilde{O}((\frac{\Stot^{\omega+1} + \Stot^2\Atot}{1-\gamma})\log(\frac{1}{1-\gamma}))$ time.
  \label[corollary]{cor:random-policy-iteration}
\end{corollary}

\begin{proof}
    We show that \cref{alg:randomized-policy-iteration} is an instantiation of \cref{alg:reduction}, where $\pi$ in \Cref{line:policy} is chosen uniformly at random form the remaining policies (as in \Cref{sec:rand_reduction}) and the for loop in \Cref{line:for_rand_PI} is used to implement \Cref{line:aprox_val_solve} in \Cref{alg:reduction}, i.e., compute $\epsilon$-optimal values $v$ for $\epsilon = \slack_{\max}^{\M}(v_{\pi}) \cdot \frac{1-\gamma}{3(1+\gamma)}$. 
    For this, note that 
    \Cref{lem:one_step_policy_iter} implies that after $\lceil \frac{\ln (3(1+\gamma)/(1-\gamma)^2)}{1-\gamma}\rceil$ iterations of policy iteration---i.e., $\lceil \frac{\ln (3(1+\gamma)/(1-\gamma)^2)}{1-\gamma}\rceil$ iterations of the for loop in \Cref{line:for_rand_PI}---starting from $\pi$, the new policy $\pi'$ satisfies that $\|v^\M_{\pi'} - v^\M_*\|_{\infty} \le \frac{(1-\gamma)^2}{3 (1+\gamma)} \|v^\M_{\pi} - v^\M_*\|_{\infty}$. 
    We now show that $(1-\gamma) \|v^\M_{\pi} - v^\M_*\|_{\infty} \le \slack_{\max}^{\M}(v_{\pi})$. For this, note that it suffices to show $(1-\gamma) \|v_{*}^{\M^\pi}\|_{\infty} \le \slack_{\max}^{\M}(v_{\pi})$. This follows immediately from the fact that $\slack_{\max}^{\M}(v_{\pi})$ is the maximum reward in the $\gamma$-DMDP $\M^\pi$. 
    Therefore, the loop in \Cref{line:for_rand_PI} computes $\epsilon$-optimal values $v$ for $\epsilon = \slack_{\max}^{\M}(v_{\pi}) \cdot \frac{1-\gamma}{3(1+\gamma)}$. Thus, the correctness of the algorithm follows from \Cref{thm:framework_correctness}. 

    To bound the running time, note that computing $v_{\pi}$ can be implemented in $O(\Stot^\omega)$ time (\Cref{fact:compute-value}). Additionally, $\slack_{\max}^{\M}(v_{\pi})$ and each iteration of the for loop in \Cref{line:for_rand_PI} can be implemented in $O(\Stot \Atot)$. 
    Moreover, by \Cref{lemma:potential}, the number of iterations $T$ is at most $\otilde(\Stot)$ in expectation, which implies the desired running time. %
\end{proof}

This result opens the door to analyzing a range of randomized variants of the policy iteration algorithm. 
We chose to analyze this particular algorithm in \Cref{alg:randomized-policy-iteration} as it is consistent with the framework discussed at the beginning of \Cref{sec:reduction}. 
We leave the exploration of alternative randomized variants of the policy iteration algorithm for future work.

\section*{Acknowledgements}

We thank Bento Natura for helpful comments and discussions. 
Andrei Graur was supported in part by the Nakagawa departmental fellowship award from the Management
Science and Engineering Department at Stanford University, Microsoft Research Faculty Fellowship, NSF CAREER Award CCF-1844855, and NSF Grant CCF-1955039.
Aaron Sidford was supported in part by a Microsoft Research Faculty Fellowship, NSF CAREER Award CCF-1844855, NSF Grant CCF-1955039, and a PayPal research award.
Ta-Wei Tu was supported in part by a Microsoft Research Faculty Fellowship and NSF CAREER Award CCF-1844855.

\bibliographystyle{alpha}
\bibliography{bib.bib}

\appendix
\section{From Monte Carlo Algorithms to Las Vegas Algorithms}\label{sec:monte-carlo-to-las-vegas}

In this section we prove \cref{lemma:mc-to-lv-intro}, restated below. 

\MCtoLV*

\begin{proof}
  Consider running the Monte Carlo algorithm on $\M$ with accuracy $\delta^\prime \defeq \delta \cdot \frac{1 - \gamma}{1+\gamma}$.
  Let $v$ be the values it returns.
  In $O(\Stot\Atot)$ time we can compute $\T_*(v)$ and $\|v - \T_*(v)\|_{\infty}$.
  If $\|v - \T_*(v)\|_{\infty} \leq \delta\|r\|_{\infty} \cdot (1-\gamma)$, then we have by \cref{lemma:value-bound} that $\|v - v_*\|_{\infty} \leq \delta\|r\|_{\infty}$ and thus we can output $v$.
  On the other hand, again by \cref{lemma:value-bound} we know that $\|\T_*(v) - v\|_{\infty} \leq \|\T_*(v) - v_*\|_{\infty} + \|v_* - v\|_{\infty} \leq (1+\gamma)\|v - v_*\|_{\infty}$ and thus if $v$ is $\delta^\prime \|r\|_{\infty}$-optimal then we must have that $\|v - \T_*(v)\|_{\infty} \leq \delta\|r\|_{\infty} \cdot (1-\gamma)$.
  Thus, with constant probability this procedure  outputs values $v$ with certified  optimality.
  Repeating this procedure until such a $v$ is output yields the result.
\end{proof}

\end{document}